%% file: RegulatorCareerOption_rev.tex
\documentclass[11pt]{article}
\usepackage[dvipsnames,table]{xcolor} 
\usepackage{setspace}            
\usepackage{amssymb,amsmath, amsthm, setspace, color, fullpage}
\usepackage{mathrsfs}            
\usepackage{dsfont}              
\usepackage[margin=1in]{geometry} 
\usepackage{graphicx}
\usepackage{tikz}
\usetikzlibrary{arrows.meta,positioning,calc,fit,shapes.geometric}
\usepackage[english]{babel}
\usepackage{float}
\usepackage[justification=centering]{caption}
\usepackage{subcaption}
\usepackage{placeins}
\usepackage{lscape} 
\usepackage{wrapfig}
\usepackage{lipsum}
\usepackage[authoryear]{natbib}
\usepackage[titletoc,toc,page]{appendix}
\definecolor{electcyan}{rgb}{0.0,1.0,1.0}
\definecolor{darkcyan}{rgb}{0.0,0.55,1.0}
\definecolor{cyanproc}{rgb}{0.0,0.72,0.92}
\definecolor{brightgreen}{rgb}{0.4,1.0,0.0}
\definecolor{Ao}{rgb}{0.0,0.0,1.0}  
\definecolor{darkblue}{rgb}{0.0,0.0,0.55} 
\definecolor{electcrimson}{rgb}{1.0,0.0,0.25} 
\definecolor{americanrose}{rgb}{1.0,0.01,0.24}
\usepackage[pdfborder={0 0 0},final=true,breaklinks = true,colorlinks=true,linkcolor=red,citecolor=blue]{hyperref}
\usepackage{calc}     
\usepackage[shortlabels]{enumitem} 
\usepackage[normalem]{ulem} 
\usepackage[capitalize,nameinlink,noabbrev]{cleveref}
\crefname{thm}{theorem}{theorems}
\Crefname{thm}{Theorem}{Theorems}
\crefname{prop}{proposition}{propositions}
\Crefname{prop}{Proposition}{Propositions}
\crefname{lem}{lemma}{lemmas}
\Crefname{lem}{Lemma}{Lemmas}
\crefname{cor}{corollary}{corollaries}
\Crefname{cor}{Corollary}{Corollaries}
\crefname{defn}{definition}{definitions}
\Crefname{defn}{Definition}{Definitions}
\crefname{rem}{remark}{remarks}
\Crefname{rem}{Remark}{Remarks}
\crefname{assumption}{assumption}{assumptions}
\Crefname{assumption}{Assumption}{Assumptions}
\numberwithin{equation}{section}
\newcommand{\tab}{\hspace*{2em}}              
\theoremstyle{plain}
\newtheorem{thm}{Theorem}[section]

\newtheorem{lem}[thm]{Lemma}
\newtheorem{prop}[thm]{Proposition}
 
\theoremstyle{plain}
\newtheorem{assumption}[thm]{Assumption}      
\theoremstyle{definition}

\theoremstyle{definition}
\newtheorem{defn}{Definition}[section]
\theoremstyle{remark}
\newtheorem{rem}{Remark}[section]

\renewcommand{\thefootnote}{\fnsymbol{footnote}} 

\makeatletter
\let\xx@thm\@thm
\AtBeginDocument{\let\@thm\xx@thm}
\makeatother

\begin{document}
\pagenumbering{roman}
\title{\textbf{A Regulator's Career Option: Revolving Doors, Regulatory Signals, and Firm Tail Risk}
\footnote{This paper benefited from the detailed comments of Samuel Myers, Jr. and Gary Dymski. We thank seminar participants at the University of Leicester and conference participants at the Academy for Behavioral Finance \& Economics (Robert B. Durand, discussant), the Southern Finance Association (Rachel Graefe-Anderson, discussant), the ESSA Biennial Meeting, and the Midlands Economic Theory and Applications Workshop. We are also grateful to Paul Kupiec, Leonard Nakamura, Doug Robertson, Ernesto Dal B\'o, Oliver Martin, Tony Leiman, Linwood Tauheed, Michael Crew, Rebel Cole, Asha Sundaram, Peter Hammond, and Haim Abraham for their valuable comments and advice. Any errors which may remain are ours. We gratefully acknowledge the research support of the Institute for Innovation and Technology Management.}
}
\author{
    G. Charles-Cadogan
     \thanks{University of Leicester, School of Business, Leicester, LE2 1RQ;  Tel: +44 (0)116 229 7385; e-mail: \textcolor[rgb]{0.00,0.00,1.00}{\href{mailto:gocadog@gmail.com}{gocadog@gmail.com}}}~\\~\\ Working Paper
    }
\date{\today}
\maketitle
\thispagestyle{empty}
\begin{abstract}
\noindent This paper develops a model of revolving doors in which a regulator designs policy signals that affect a regulated firm's capital structure while holding an implicit career option to leave public service and join the regulated firm or its advisers. The model has three linked mechanisms. First, regulatory signals alter the firm's admissible pricing and leverage choices by changing the tradeoff between consumer surplus and firm profit. Second, because the regulator's outside payoff is convex in the value of the regulated firm, the value of the regulator's career option increases with firm value, leverage, and signal-induced volatility. Third, when signal-induced volatility raises the option vega of firm equity, the same mechanism can increase tail risk and lower distance to default. A proof-of-concept calibration illustrates the model's leverage-vega warning regions and shows how a regulator-induced signal can move a firm from a go-leverage region into an early-warning or crash-risk region. The calibration is extended with maturity comparisons, a sensitivity grid over compensation convexity and initial leverage, a bank-like parameter anchor, stochastic signal simulations, and a stopping-horizon diagnostic. External evidence from federal payroll data on post-employment restrictions supports the model's career-option primitive, while the exercises also organize historical regulatory failures: Lehman Brothers as a regulatory-vacuum case, Silicon Valley Bank as a distorted-signal and weak-supervision case, and Signature Bank as a converted-signal case in which post-public-service incentives and lighter oversight coincided with rapid growth, uninsured-deposit exposure, and tail risk. The model does not predict collapse dates; it helps explain and predict when regulatory forbearance, weak debt penalties, and career-option incentives can delay binding oversight until firms cross a leverage-vega warning threshold.
\end{abstract}
\newpage
\tableofcontents
\listoffigures
\listoftables
\clearpage
\pagenumbering{arabic}
\renewcommand\thefootnote{\fnsymbol{footnote}}
\renewcommand\thefootnote{\arabic{footnote}}
\section{Introduction}\label{sec:Intro}
\tab This paper studies a revolving-door problem in which a regulator can affect a firm's capital structure while contemplating a future move to that firm, a lobbying organization, or an affiliated consultancy. The model is not a claim that all regulators behave opportunistically. It is a mechanism-design and contingent-claim model of a particular conflict: when the regulator's future labor-market payoff is convex in the regulated firm's value, regulatory discretion can create incentives to favor signals that raise firm value, leverage, and option-like compensation.

\tab The players are a regulator, a regulated firm, and the public. The regulator chooses a regulatory signal $\theta$ or, more generally, a state-dependent signal $s(\omega,\theta)$ that affects the firm's admissible strategy space. The firm observes or infers the signal and chooses pricing, leverage, or investment actions from the induced strategy set. The public receives consumer surplus and bears part of the cost of distorted regulation. The regulator maximizes a value function that includes consumer surplus, firm profit, and the private value of a career option to quit and join the regulated firm. The firm maximizes profit subject to the strategic environment induced by the regulator's signal. The model therefore has two connected incentive problems: the public-sector welfare problem and the regulator's private option-value problem.

\tab The central object is the regulator's implicit career option. If the regulator can leave public office and receive an equity-linked payoff from the regulated firm, then the payoff resembles a call option with value $(\alpha V-K)^+$, where $V$ is firm value, $K$ is the regulator's current compensation or outside strike value, and $0<\alpha<1$ is the share of firm value embedded in the future compensation package. Regulatory signals matter because they affect $V$, leverage $D/V$, and signal-induced volatility $\sigma(s(\theta))$. A signal that raises firm value or volatility can therefore increase the value and vega of the regulator's career option even when it also raises firm tail risk.

\tab The paper delivers three theoretical claims. First, regulatory signals embedded in firm capital structure can divide the firm's strategy space into consumer-surplus and firm-profit regimes. Second, the regulator's career option is increasing in firm value and is more valuable for levered firms when the regulator's future payoff is tied to equity-like compensation. Third, under an information-based option-pricing interpretation, signal-induced volatility raises the vega of the option on firm value and can generate a leverage-vega warning region in which firm value-at-risk is higher than it would be without regulatory capture.

\tab The calibration does not estimate a causal effect or assert that a particular industry episode followed the model. It illustrates how the model's primitives generate the qualitative regions implied by the propositions: a go-leverage region, an early-warning region, and a crash or bankruptcy region. The extended exercises then show how those regions respond to maturity, leverage, compensation convexity, bank-like parameter anchors, stochastic signal realizations, and quitting-horizon choices. This structure also provides a diagnostic language for prominent financial failures in which delayed, missing, or distorted regulatory signals allowed leverage, duration risk, uninsured-deposit dependence, or other tail exposures to accumulate before binding intervention occurred.

\tab The paper relates first to classic theories of revolving doors and regulatory capture \citep{Che1995,Salant1995,Stigler1971}. Recent work has sharpened this literature by distinguishing several mechanisms. \citet{DeChiaraSchwarz2021} develop a dynamic theory in which a future private-sector job can dominate hidden transfers as a capture device because the revolving door publicly signals the firm's willingness to reward lenient regulators. \citet{Kempf2020JobRating} shows that revolving-door incentives can also operate through reputational and accuracy incentives in credit-rating labor markets, which is an important counterweight to one-sided capture interpretations. \citet{KaticKim2024} show that regulator-to-firm mobility can improve regulatory outcomes for firms, but that the advantage is concentrated just before the mobility event. \citet{EmeryFaccio2025ExecutiveBranch} map executive-branch revolving-door appointments into corporate positions and show that firms are more likely to receive procurement contracts after appointing recent former regulators, while some contracts become more costly through renegotiation. \citet{YatesCardinTrudeau2021} emphasize ``lobbying from within'' and entry-side cultural capture, while \citet{HeimsMoxon2024} argue that capture claims require evidence on the mechanism by which private influence changes regulatory decisions. These recent papers motivate the present model's focus on the timing and mechanism of the regulator's private payoff rather than on a blanket assertion of capture.

\tab The paper also connects to regulation and capital structure \citep{SpiegelSpulber1994,BortolottiCambiniRondiSpiegel2011}, political connections and financial firms \citep{BraunRaddatz2010,LucaSeruTrebbi2014,Lambert2017}, and option-based interpretations of firm equity and compensation \citep{Merton1974,Hull2006}. Recent financial-regulation evidence is especially close to the paper's mechanism. \citet{HirtleKovner2022} review the supervisory literature and emphasize that supervision uses soft information to discipline bank risk-taking. \citet{KimPlosserSantos2018Leveraged} show that leveraged-lending guidance affected closely supervised banks but also shifted risk toward nonbanks, illustrating why an apparently binding regulatory signal can move risk rather than eliminate it. \citet{ColonnelloKoetterSclipWagner2023} study reverse revolving doors in European bank supervision and report that former financial-industry executives on supervisory boards are associated with positive bank valuation effects, lower regulatory capital, and faster bank growth before the Single Supervisory Mechanism. \citet{HaselmannSchoenherrVig2018Elite} show how elite social networks can distort bank credit allocation, reinforcing the paper's emphasis on network-mediated regulatory and financing advantages. \citet{BerchtoldBowlerBrockmanLiTresl2026PED} find that politically experienced directors reduce the sensitivity of corporate investment to policy uncertainty, especially for irreversible investment, which supports the model's claim that political and regulatory human capital can change firm policy under uncertainty. \citet{SchramUle2024} show experimentally that regulatory independence can reduce electoral holdup but can also entrench capture when policymakers' objectives align with regulated owners. The contribution here is to connect these strands in a compact mechanism: a regulator's policy signal changes firm strategy and capital structure, while the regulator's future labor-market payoff has the convexity of an option on the regulated firm.

\tab Direct empirical evidence on federal employees supports the model's career-option primitive. \citet{KalmenovitzVijXiao2025} use payroll records on 22 million federal employees and the salary threshold that triggers post-employment restrictions under 18 U.S.C. \S 207 to show significant bunching just below the restriction cutoff. In agencies with broad regulatory powers, limited political supervision, and frequent interaction with high-paying industries, they estimate that 49\% of regulators respond to revolving-door incentives and sacrifice 7.4\% of wage potential to preserve private-sector opportunities. They also find that revolving regulators initiate fewer enforcement actions and issue rules with lower compliance costs. In the language of the present model, their evidence verifies the maintained economic object: the outside job opportunity is valuable before exit, regulators trade current public compensation against future private-sector option value, and the incentive is associated with softer regulatory signals.

\tab The rest of the paper proceeds as follows. In \cref{sec:Model} we introduce the model and derive the regulator's career-option mechanism. In \cref{sec:ProofOfConcept} we provide a proof-of-concept calibration of leverage, vega, and warning regions. In \cref{sec:Conclusion} we conclude. The appendix states the abstract legislative mechanism that supports the regulator's signalling problem.

\section{The Model}\label{sec:Model}
\tab The model has three primitive objects: a regulator type $\theta\in\Theta$, a state of nature $\omega\in\Omega$, and a regulatory signal $s(\omega,\theta)$ that maps the regulator's information and type into a strategy space $G$ for the firm. The firm responds to the induced strategy space by choosing prices, leverage, or investment policies. The regulator values public welfare, but also holds an implicit career option $C:G\rightarrow X$ whose payoff is increasing in the regulated firm's value. Thus regulatory signals affect both the firm's strategy set and the regulator's private option value. The propositions below first formalize the abstract career option, then show how welfare tradeoffs, leverage, option value, and vega interact.
\begin{figure}[!htbp]
\centering
\resizebox{\linewidth}{!}{%
\begin{tikzpicture}[
  node distance=1.25cm and 1.55cm,
  box/.style={draw=black, rounded corners=2pt, align=center, minimum width=3.0cm, minimum height=0.95cm, fill=gray!10},
  smallbox/.style={draw=black, rounded corners=2pt, align=center, minimum width=2.55cm, minimum height=0.80cm, fill=gray!6},
  arrow/.style={-{Latex[length=2.2mm]}, thick, draw=black},
  feedback/.style={-{Latex[length=2.2mm]}, thick, draw=black, dashed}
]
\node[box] (reg) {Regulator\\ type $\theta$};
\node[box, right=of reg] (sig) {Policy signal\\ $s(\omega,\theta)$};
\node[box, right=of sig] (strat) {Firm strategy set\\ $G(s)$};
\node[box, right=of strat] (firm) {Firm choices\\ $p,D,V,\sigma$};
\node[smallbox, below=of strat] (public) {Public payoff\\ $CS(p)$};
\node[smallbox, below=of firm] (profit) {Firm profit\\ $\Pi(p,V,D)$};
\node[box, below=of sig] (option) {Career option\\ $C(s)=(\alpha V-K)^+$};
\draw[arrow] (reg) -- (sig);
\draw[arrow] (sig) -- (strat);
\draw[arrow] (strat) -- (firm);
\draw[arrow] (firm) -- (profit);
\draw[arrow] (firm) |- (public);
\draw[arrow] (profit) -- (option);
\draw[feedback] (option) -- (reg);
\draw[feedback] (public) -- (reg);
\node[draw=black, dashed, rounded corners=2pt, fit=(sig)(strat)(firm)(profit)(option), inner sep=0.25cm, label={[align=center]above:Signal-induced capital-structure channel}] {};
\end{tikzpicture}%
}%
\caption{Mechanism sketch: regulatory signal, firm response, and career-option feedback}\label{fig:TikZMechanismSketch}
\begin{minipage}{0.90\linewidth}
\footnotesize
\textit{Notes:} The regulator's signal changes the firm's feasible strategy set. Firm choices affect consumer surplus, firm profit, firm value, leverage, and volatility. Because the regulator's outside payoff is convex in firm value, the career option feeds back into the regulator's signal incentives.
\end{minipage}
\end{figure}
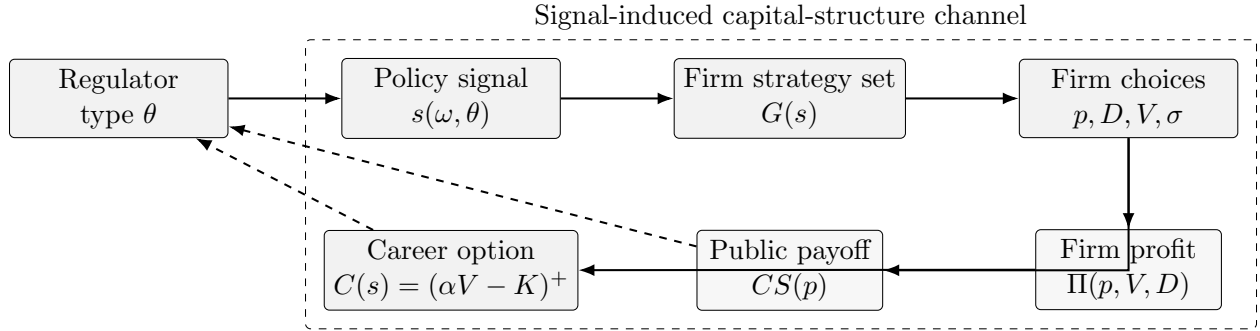
\begin{prop}[Regulator's abstract career option on regulated firms]\label{prop:RegulatorAbstractOption}
   Let $\Omega$ be a sample space for states of nature, $\Theta$ be a type space for regulators, $G$ a strategy space, and $\mathcal{M}=(G,C)$ be an abstract legislation mechanism in \Cref{defn:Mechanism}, and $s$ be a regulatory signalling function such that $s:\Omega\times\Theta\rightarrow G$. Let $f$ signify the regulator's unobservable choice function. So that $f:\Omega\times\Theta\rightarrow X$ is a stochastic choice function mapped into outcome space. Specifically, let $C:G\longrightarrow X$ be the regulator's career [call] option on regulated firms. We claim that the abstract value of the option is given by the legislative mechanism: $$f(\omega,\theta)\cong (C\circ s)(\omega,\theta)$$
   where $\cong$ represents an isomorphism.
\end{prop}
\begin{proof}
   See \cref{apx:AbstractMechanismDesign}.
\end{proof}
\begin{rem}
   Technically, the legislative mechanism is supposed to induce the  firm to engage in some action for the regulator's benefit. As we will see in the sequel, that ``something`` is profit maximization or leverage. \cref{subsec:WelfarEfffectsOfRegulateSignal} illustrates application of elements of this abstract proposition,\hfill $\Box$
\end{rem}
\subsection{Welfare effects of embedding regulator signals in capital structure}\label{subsec:WelfarEfffectsOfRegulateSignal}
Without loss of generality, for the purpose of exposition, in this subsection our analysis is based on deterministic parameterizations. Furthermore, we assume that the form is monopolistically competitive so it can execute a pricing strategy. Here, the regulator wants to maximize her social welfare function, i.e. value function, comprised of consumer surplus $(CS)$ and the firm profits $(\Pi)$ [s]he affects by embedding regulatory signals $\theta$ in the firm's capital structure. That is, through choices and regulatory decisions the regulator attempts to affect firm equity values. Our motivation for the welfare function is based on \citet[pg.~427]{SpiegelSpulber1994}, who state:
\begin{quote}
     [T]he  regulator  establishes the  regulated  price  by  maximizing  a  welfare function,  taking  the  firm's  investment  and  capital  structure  as  given.  The  balancing  of consumers'  and investors'  interests  is made explicit by  assuming that the welfare  function is  a  weighted  sum of  consumers'  surplus  and  firm profits.  Then,  the  regulated  firm  produces  its  output  and  the  regulated  market  clears.  The  focus  of  this  article  is  on  characterizing  the  subgame equilibrium of  the  regulatory  game.
\end{quote}
See also \citet[pg.~916]{BaronMyerson1982}. However, in our model capital structure is embedded with regulatory signals and our equilibrium is based on a reversed game theory, i.e. mechanism design approach.\footnote{We thank Gary Dymski for motivating clarification of this point.}  This is a source of ``market friction" consistent with the tradeoff hypothesis of capital structure theory, i.e., ``tradeoff theory looks at how businesses balance the pros and cons of different forms of financing"  \citep{Kwan2009}. We show below, that that is indeed the case. For example, as regulators tradeoff welfare for firm profit, the price strategy space is affected. However, the latter is tied to the firms activity ratio (e.g., sales/(total assets)) which is a key predictor of bankruptcy \citep[e.g.,][]{Altman1984} in tradeoff theory otherwise known as the static theory of capital structure which examines tradeoff between corporate tax shields and expected bankruptcy costs of leverage.

\tab Consistent with \citet[pg.~429,~equation~(6)]{SpiegelSpulber1994}, we identify the following primitives used to construct our criterion function.\par\noindent\textbf{Primitives}
\begin{description}[labelwidth=\widthof{\bfseries D~~~~}]
    \item[$p\;\triangleq$]  price of firm output
    \item[$k\;\triangleq$]  firm capital investment
    \item[$D\triangleq$]  nominal amount of debt owed to bondholders
    \item[$E\triangleq$]  market value of the firm equity
    \item[$\rho\triangleq$] debt/distress penalty weight in the regulator value function
\end{description}
Unlike \citet[pg.~428]{SpiegelSpulber1994}, who let $E$ be the market value of new equity, and $B$ be the market value of bonds, we let the firm's capital $k$ be its value $V(\theta)$ which is a combination of debt $D(\theta)$, and market value of equity $E(\theta)$. In that way, the market value of equity is transmitted through firm value for a given amount of debt. Assuming no transfers or taxes, \citet[pg.~429,~equation~(6)]{SpiegelSpulber1994} posit a ``utilitarian welfare function" which we modify to construe as our regulator value function written as
\begin{align}
   W(p,V(\theta),D(\theta)) &= CS(p)+b\Pi(p,V(\theta),D(\theta))-\rho D(\theta)\label{eq:RegulatorValueFunction}\\
   V(\theta) &= D(\theta)+E(\theta)\label{eq:RegulatorSignalEmbededCapStructure}
\end{align}
where $CS(p)$ is consumer surplus; $0<b<1$ is a welfare weight; $\Pi(p,V(\theta),D(\theta))$ is firm profits; and $\rho>0$ is the regulator's penalty weight on debt or expected distress exposure.
\begin{assumption}\label{assum:RegValueFunction}
   \begin{enumerate}[{A}.1]
      \item  Debt is signal dependent, $D_\theta\geq 0$. In the high-leverage region, $L(\theta)=D(\theta)/V(\theta)$ satisfies $L_\theta>0$, equivalently $D_\theta/D>V_\theta/V$.
      \item  Regulator signals favor firm capital and profits through the firm-value and leverage channels, i.e. $\Pi_V>0$, $V_\theta>0$, and $b\Pi_VV_\theta+(b\Pi_D-\rho)D_\theta>0$ over the relevant region. This includes the case in which debt or distress is weakly penalized, so that leverage-taking is not fully internalized.
      \item  Consumer surplus decreases in price, $CS_p<0$.
   \end{enumerate}
\hfill $\Box$
\end{assumption}
\begin{figure}[!htbp]
\centering
   \begin{minipage}[h]{0.95\linewidth}
      \captionof{figure}{Impact of regulator signal $\theta$ on price $(p)$ and firm profits $\Pi$ }
      \centerline{\includegraphics[scale=0.6,height=11.5cm,width=11cm]{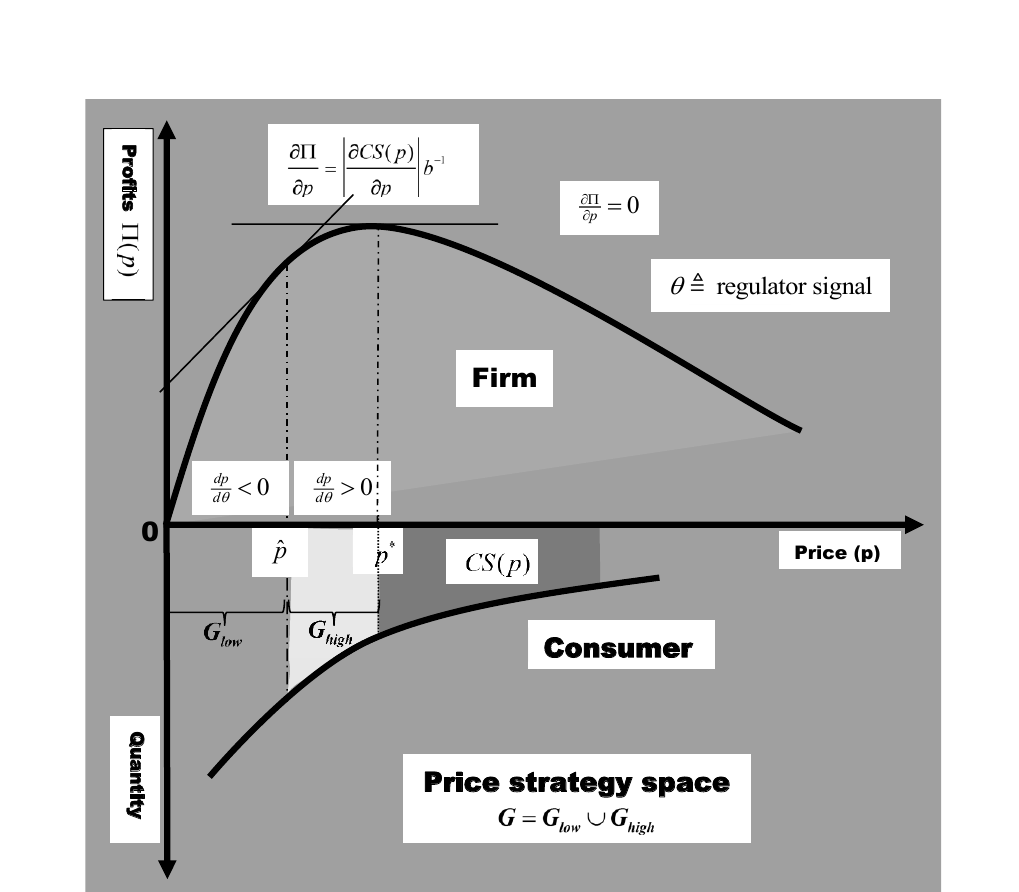}}
      \label{fig:RegulatorCareerOptionFig2}
       \begin{flushleft}
         In regulator induced price strategy space $G_{low}=\{0<p\leq \hat{p}\}$ the rate of price change $p_\theta$ relative to regulator signal $\theta$ is decreasing. For $G_{high}=\{\hat{p}<p\leq p^*\}$ the rate of price change is increasing (and consumer surplus is falling), and firm profits are ultimately maximized at $p^*$. An admissible [price] strategy space for consumer welfare firm profit tradeoffs is the support $G=G_{low}\cup G_{high}$.
       \end{flushleft}
   \end{minipage}
\end{figure}
\noindent Under \Cref{assum:RegValueFunction}, the regulator signal is channeled through firm value and debt. The effective total differential of \eqref{eq:RegulatorValueFunction} is given by
\begin{align}
    d(W(p,V(\theta),D(\theta))) &= CS_p\,dp + b\{\Pi_p\,dp+\Pi_VV_\theta\,d\theta+\Pi_DD_\theta\,d\theta\}-\rho D_\theta\,d\theta .
\end{align}
Equate the left hand side to $0$, and let subscripts denote partial derivatives. Then
\begin{align}
    (CS_p+b\Pi_p)dp + \{b\Pi_VV_\theta+(b\Pi_D-\rho)D_\theta\}d\theta &=0,\nonumber\\
    \frac{dp}{d\theta} &=p_\theta=\frac{b\Pi_VV_\theta+(b\Pi_D-\rho)D_\theta}{|CS_p|-b\Pi_p}.\label{eq:RegulateFirstOrderpriceEffects}
\end{align}
The term $(b\Pi_D-\rho)D_\theta$ is the incremental debt channel. If debt or distress is weakly penalized relative to its private profit contribution, this term amplifies the price-signal response. If $\rho$ is large, it dampens or can reverse the response.

\tab The price-signal ratio induced by regulator signal $\theta$ in \eqref{eq:RegulateFirstOrderpriceEffects}, informs the firm about admissible pricing strategies, leverage, and profits depicted in \cref{fig:RegulatorCareerOptionFig2}. There, in the price range $G_{low}=\{0<p\leq \hat{p}\}$ the rate of price change $p_\theta$ relative to $\theta$ is decreasing. This implies that consumer surplus is increasing. However, firm profits are not being maximized. In the range $G_{high}=\{\hat{p}<p\leq p^*\}$ the rate of price change is increasing (and consumer surplus is falling), and firm profits are ultimately maximized at $p^*$. At which point the darker shaded region is the loss in consumer surplus $CS(p)$ or consumer welfare sacrificed for firm profits. Thus, regulator signals allow the firm to adopt a pricing strategy in either one of the price ranges according to whether the regulator's focus is on maximizing consumer welfare at the expense of firm profits or vice versa.\footnote{\citet{ZimperHassan2012} introduced a model in which regulators learn about firms' price collusion strategies through dynamic price dividend ratios.} In other words, an admissible [price] strategy space is the support $G=G_{low}\cup G_{high}$.

\tab The same signal also has a capital-structure implication. Since $L(\theta)=D(\theta)/V(\theta)$,
\begin{align}
   L_\theta=\frac{D_\theta V(\theta)-D(\theta)V_\theta}{V(\theta)^2}.\label{eq:LeverageSignalDerivative}
\end{align}
Thus leverage rises when debt grows faster than firm value, $D_\theta/D>V_\theta/V$. This is the relevant revolving-door mechanism: if debt-taking is not heavily penalized in the regulator value function, the firm can move toward a highly levered region. Higher leverage can increase the convex value of equity-linked career compensation, but it also reduces the distance-to-default component $\log(V(\theta)/D(\theta))$ in the option and default calculations below.

\tab Since the value of the regulator's career option is tied to firm profit, we should expect regulators who plan to exercise that option to send signals that support pricing in the range $G_{high}$.  The revolving-door literature provides examples consistent with this mechanism. \citet[pp.~217-218]{DalBo2006} report that former FCC commissioners who take industry employment ``increase their support for industry interest by nearly 11 percent during their last year in office". In that case, consumer welfare is sacrificed for regulator career goals. See, e.g., \citet[pg.~8]{Kane2011}, who summarizes a related issue thus:  ``The importance of political, bureaucratic, and \emph{career interests in regulatory decision- making} allows such firms to screen regulatory appointments and to distort regulatory policies \emph{ex ante} and to reshape their enforcement \emph{ex post}" (emphasis added). We summarize these results in
\begin{prop}[Regulator's career goal and tradeoff between consumer welfare and firm profit]\label{prop:RegCareerGoalFirmProfit}
    Assume that there are no transfer payments. Let $\theta$ be a regulator signal, and $W(p,V(\theta),D(\theta)) = CS(p)+b\Pi(p,V(\theta),D(\theta))-\rho D(\theta)$ be the regulator's value function, where $p$ is the price of the firm's product, $CS(p)$ is consumer surplus, $\Pi$ is a profit function as indicated, $V(\theta)$ is firm value with embedded regulator signal, $D(\theta)$ is debt, and $\rho>0$ is the debt/distress penalty weight. If \Cref{assum:RegValueFunction} holds, then regulator signals induce dichotomous decreasing and increasing price regimes in strategy space $G=\{G_{low}\cup G_{high}\}$, where $G_{low}=\{(0,\hat{p}]\}$ and $G_{high}=\{(\hat{p},p^*]\}$. Moreover, regulators will trade consumer welfare for firm profits in increasing price regimes $G_{high}$ in order to enhance the value of their career option $C:G\longrightarrow X$ to quit and join firm management. \hfill $\Box$
\end{prop}
In the proof-of-concept exercise below, the high-price region is interpreted abstractly as a firm-profit region.

\subsection{Stock options with regulatory signalling}\label{subsec:StockOptionsWithRegSignals}
In this section, we apply \Cref{prop:RegulatorAbstractOption} in the context of what we now know about the firms capital structure embedded with regulator signals. Where possible we suppress the $\omega$ notation since it is acknowledged by the probability distributions in the formulae that follow.  We retain \Cref{assum:RegValueFunction} as it pertains to signal-dependent firm value and debt. With few exceptions, the assumptions below are consistent with \citet[pg.~450]{Merton1974}.

\tab We note in passing that regulator signals are a source of market friction that affects capital structure. Since the firm responds to the signal(s) it receives from the regulator, let $V(s(\theta))$ be the value of the firm. The notation here is slightly more general than \eqref{eq:RegulatorSignalEmbededCapStructure}. Suppose that the regulator's potential ``managerial compensation" with the firm would be the equity stake $\alpha V,\;0<\alpha <1$, upon a career change. If $K\in X$ is the ``value" of the regulators compensation package in her current position as regulator, then [s]he may want to exercise a career [call] option at time $t$. The abstract topological lifting in \cref{fig:RegulatorCareerOptionFig01} in \cref{apx:AbstractMechanismDesign} shows that (suppressing $\omega$) the choice function at time $t$ reduces to
\begin{equation}
   f(t;\theta)=(\alpha V(t,s(\theta))-K,0)^+
\end{equation}
More formally, that abstract [call] option compensation scheme takes the form
\begin{align}
   C(t;s(\theta^*)) &=\max_\theta(\alpha V(t;s(\theta))-K,0)^+\label{eq:RegulatorCareerCallOption}
   \intertext{or in the alternative}
   C(\tau^*;s(\theta)) &=\sup_{\tau\in\mathcal{T},\theta\in\Theta}(\alpha V(\tau,s(\theta))-K,0)^+\label{eq:RegAbstractOption}
\end{align}
where, $\theta^*$ is the type of regulator who maximizes her option, $\tau\in\mathcal{T}$ is a stopping time and $\tau^*$ is the optimal stopping time in the set $\mathcal{T}$ of all possible quitting times for the regulator. The foregoing is summarized in the following
\begin{lem}[Regulator's career call option]\label{lem:RegulatorCareerCallOption}
   The regulator's career call option $C(t,s(\theta))$ increases when embedded regulatory signals increases firm value $V(t,s(\theta))$ and it is evaluated as $C(\tau^*;s(\theta^*)) =\sup_{\tau\in\mathcal{T}}\max_{\theta\in\Theta}(\alpha V(\tau,s(\theta))-K,0)^+$. \qed
\end{lem}
\noindent Of course $V(t,s(\theta)|\mathcal{F}_t)$ is the value of the firm at time $t$ given the ``filtration" of information $\mathcal{F}_t$, and it has to be discounted from some ``terminal date" $T$.  To keep it simple, assume that the option is a European type\footnote{For the purpose of this paper, the added complexity of American style option pricing may cause us to loose sight of the forest for the trees. See, e.g., \citet{Myneni1992} for a taxonomy of American option pricing.} so that it will be exercised at a ``known" terminal date, and that the strike price is $K$. The regulator exercises a call option based on her type with the premium as indicated. It should be noted in passing that on the other side of the regulator's call option contract is a third party guarantor who holds an offsetting put contract. Namely, the regulating agency or taxpayers. See, e.g., \citet{Merton1977}. Since firm value is reflected by the price of the firms stocks, the regulator's action is functionally equivalent to a stock option.\footnote{See \citet{HallMurphy2000}.} In this case, the stock price reflects the signals from the regulator. More on point, from \eqref{eq:RegulatorSignalEmbededCapStructure} we find (suppressing $t$) that the firm's equity is given by:
\begin{align}
   E(s(\theta)) &= V(s(\theta))-D(s(\theta))\\
   \intertext{So the regulator's virtual call option is based on}
   \alpha E(s(\theta))^+ &= \max\{\alpha V(s(\theta))-\alpha D(s(\theta)),0\}\label{eq:RegulatorCallOptionOnFirmEquity}\\
   \intertext{This implies that the strike price in \eqref{eq:RegulatorCareerCallOption} based on the value of the regulator's pre-industry compensation package, must be compared to the firm's debt before [s]he exercises her option. Specifically, we must consider whether}
   K-\alpha D(s(\theta))\lessgtr 0\label{eq:RegulatorPreIndustryPackageCompare}
\end{align}
This implies that we also need a model of dynamic stock prices with information signalling. That issue is addressed in the sequel.

\subsection{Valuation of regulator career option and capital structure of the firm}\label{subsec:RegulatorOptionCapitalStructure}
Consistent with \eqref{eq:RegulatorCallOptionOnFirmEquity}, \citet[pp.~453-454]{Merton1974} showed that the option value of the equity in a firm with value $V$ and debt $D$ is given by
\begin{subequations}\label{eqgrp:FirmOptionValue}
  \begin{align}
     f(V,\tau) &= V\Phi(x_1)-D\exp(-r\tau)\Phi(x_2)\label{eqgrp:FirmOptionValueEqu}\\
     x_1 &= \frac{\log\tfrac{V}{D}+(r+\tfrac{\sigma^2}{2})\tau}{\sigma\sqrt{\tau}}\label{eq:FirmOptionXOne}\\
     x_2 &= x_1-\sigma\sqrt{\tau}\label{eq:FirmOptionXTwo}
  \end{align}
\end{subequations}
where $f(V,\tau)$ is option value of firm equity, $r$ is a risk free rate, $\tau=T-t$ for exercise date $T$ for a European call option priced at time $t$; and $\Phi$ is the cumulative normal. From the outset we note that
\begin{align}
   \alpha f(V,\tau) &= \alpha V\Phi(x_1)-\alpha D\exp(-r\tau)\Phi(x_2)\\
   \intertext{The right hand side of that equation is isomorphic to a call option}
   C(\cdot) &=(\alpha V-\alpha D,0)^+\\
   \intertext{So the regulator's valuation rests on comparing his current compensation package $K$ to the psuedo-strike price $\alpha D$ that depends on the firm's debt. That is, the exercise of her career option rests on the evaluation}
   (\alpha V-K,0)^+&\gtrless(\alpha V-\alpha D,0)^+
\end{align}
Here, we assume that the regulator knows $V, K$ and $\alpha$ but she does not know what $D$ will be on the strike date. Thus, in \eqref{eq:RegulatorPreIndustryPackageCompare} if the regulator believes that $K<\alpha D$ so that $(\alpha V-K,0)^+ > (\alpha V-\alpha D,0)^+$, then her valuation of the firm is greater than its actual value. Arguably, at this stage the regulator would be interested in legislation that induces the firm to report its debt truthfully. So she would be inclined to exercise her career option to quit and join the firm or she would stay otherwise. What is more, the regulator now has incentive to design a legislative mechanism that increases firm leverage. So she looks more favorably at levered firms because a highly levered firm increases the value of her career option. This prediction is consistent with the theoretical possibility that privately controlled firms use leverage strategically to obtain better regulatory outcomes. See, e.g., \citet{BortolottiCambiniRondiSpiegel2011}. Therefore, in the sequel we assume that $K < \alpha D$.

\subsubsection{Regulator career option pricing with signalling}\label{subsubsec:RegulatorOptionPriceWithSignals}
Because $\theta$ is unobservable, its presence in firm value, i.e. stock price, implies that it may follow a hidden Markov process. To account for this possibility, we turn to the information-based stock price model of \citet{Guo2001} and \citet{Shepp2002} on the canonical probability measure space $(\Omega,\mathcal{F},\mathrm{\mathrm{F}},P)$ with filtration $\mathrm{F}$.  In the context of our model, the Shepp-Guo model is parameterized as
\begin{equation}
   \frac{dV(t,\omega)}{V(t,\omega)} =\mu(s(t,\theta))dt+\sigma(s(t,\theta))dW(t,\omega)\label{eq:HiddenMarkovSignal}
\end{equation}
where the regulatory signal $s(t,\theta)$ is a hidden Markov process. To simplify matters we assume that $s(t,\theta)=s(\theta)$. That is, the signal is not time dependent. In that case, the problem is reduced to one of an American style option where the regulator embeds a deterministic signal in the value of the firm, and she can exercise her option to join the hitherto regulated firm at any time. Thus, firm value dynamics in \eqref{eq:HiddenMarkovSignal} is given by
\begin{equation}
    \frac{dV(t,\omega)}{V(t,\omega)} =\mu(s(\theta))dt+\sigma(s(\theta))dW(t,\omega)\label{eq:NoHiddenMarkovSignal}
\end{equation}
More on point, \citet[pg.~1373]{Shepp2002} and \citet{Guo2001} introduce evidence that information asymmetry is dispositive of volatility dynamics, and the activity of insiders leads to increased volatility. In which case we assume that regulator signals increases volatility of firm value. That is
\begin{assumption}\label{assum:RegSignalVolatility}
   $\sigma(s(\theta)) > \sigma$\hfill $\Box$
\end{assumption}
\noindent This assumption is consistent with \citet{BoyleJhaKennedyTian2011}, who find that high volatility firms should use more stock options to compensate their executives. Furthermore, \citet{CadoganCole2014} show how former regulators embed signals in capital structure and use index option compensation contracts to increase compensation.  Thus, using the formulae in \eqref{eqgrp:FirmOptionValue}, we find that
\begin{subequations}\label{eqgrp:FirmOptionValueWitSignal}
   \begin{align}
      f(V(s(\theta)),\tau) &= V(s(\theta))\Phi(x_1(s(\theta)))-D(s(\theta))\exp(-r\tau)\Phi(x_2(s(\theta)))\label{eqgrp:FirmOptionValueEquWitSignal}\\
      x_1(s(\theta)) &= \frac{\log\tfrac{V(s(\theta))}{D(s(\theta))}+(r+\tfrac{\sigma(s(\theta))^2}{2})\tau}{\sigma(s(\theta))\sqrt{\tau}}\\
      x_2(s(\theta)) &= x_1(s(\theta))-\sigma(s(\theta))\sqrt{\tau}
  \end{align}
\end{subequations}
\noindent Thus we must compare the value of the option on the firm with embedded regulator signal to the value without signal. To do that we turn to option Greeks by examining the vega $(\vartheta)$ of the option on the firm.\footnote{See \citet[pg.~361]{Hull2006} for derivation of formula for vega.} That is, for the mean zero normal distribution function $\phi$, let
\begin{align}
   \vartheta(0) =\frac{\partial f(V,\tau)}{\partial\sigma}&=V\phi(x_1)\sqrt{\tau}\label{eq:VegaBS}\\
   \vartheta(s(\theta))=\frac{\partial f(V(s(\theta)),\tau)}{\partial\sigma(s(\theta))}&=V(s(\theta))\phi(x_1(s(\theta)))\sqrt{\tau}\label{eq:VegaSG}
\end{align}
According to \Cref{assum:RegSignalVolatility} information based vega is larger than plain vanilla vega.\footnote{\citet{ChesneyStrombergWagner2012} also introduced the concept of ``asset volatility vega" which is higher than plain vanilla ``equity volatility vega".} Thus,
\begin{align}
   \vartheta(s(\theta))>\vartheta(0)&\Rightarrow V(s(\theta))\phi(x_1(s(\theta)))\sqrt{\tau}>V\phi(x_1)\sqrt{\tau}\\
   \Rightarrow\frac{V(s(\theta))}{V}&>\frac{\phi(x_1)}{\phi(x_1(s(\theta)))}
\end{align}
Under \Cref{lem:RegulatorCareerCallOption} a more levered firm with embedded regulator signals enhances regulator career option so we expect $V(s(\theta))>V$. Thus, we have
\begin{align}
\dfrac{V(s(\theta))}{V}&>\max\left\{1,\,\frac{\phi(x_1)}{\phi(x_1(s(\theta)))}\right\}\\
   \intertext{The right hand side of the equation above is satisfied if}
   \frac{\phi(x_1)}{\phi(x_1(s(\theta)))}>1 &\Rightarrow \exp\Biggl(-\Biggl(\frac{x_1^2}{2}-\frac{x_1(s(\theta))^2}{2}\Biggr)\Biggr)>1\\
   &\Rightarrow -x_1(s(\theta))<-x_1 \;\;\text{and}\;\; x_1(s(\theta))>x_1
\end{align}
Following the distance-to-default interpretation of option-based credit models, $-x_1$ is an inverse distance-to-default index. The analysis below provides a theoretical explanation for why regulator-induced volatility can lower that distance and raise tail-risk exposure.
\subsubsection{Bankruptcy warning signals: Firm vega and leverage}\label{subsubsec:WarningSignalsVegaAvgLev}
At this juncture we assume the static theory of capital structure (with perpetual debt) \citep[pp.~437-439]{RossWesterfieldJrdan2017}, stated as a lemma below, in order to motivate the internal consistency of our probability of firm bankruptcy.
\begin{lem}[Tradeoff theory of capital structure]\label{lem:TradeOffTheoryCapitalStructure}
  Let $V_L$  be the value of the levered firm, and $V_U$ its unlevered value; $T_c$ be corporate tax and $D$ be the firm's debt, $r_D$ be the interest rate on debt, $b_p(\theta)$ be the probability of bankruptcy with regulator signal $\theta$, and $C$ be the cost of firm distress, thus $V_L(\theta)=V_U+T_cD-r^{-1}_Db_p(\theta)\times C$.\qed
\end{lem}
\cref{fig:RegulatorCareerOptionOptCapFig} is a typical sketch of a levered firm's value when faced with a tax shield and expected cost of firm distress under the static theory of capital structure with perpetual debt. Firm value is increasing in regulator signal, i.e., $\partial V_L/\partial\theta >0$, prior to the optimal capital structure $D^\star/E^\star$, and decreases thereafter. This implies that $\partial b_p(\theta)<0$ when firm value is increasing and $\partial b_p(\theta)/\partial\theta>0$ when firm value is decreasing. So regulator induced probability of bankruptcy is countercyclical.

The analysis that follows is based on \citet{CadoganCole2014}. The foregoing arguments show that the information based vega for the option value of the firm is dispositive of regulator incentive to leverage regulated firms. See also \citet{ArmstrongVashishtha2012}.
\begin{figure}[!htbp]
\centering
   \begin{minipage}[h]{0.95\linewidth}
      \captionof{figure}{Optimal Capital Structure Under Static Theory with Perpetual Debt}
      \centerline{\includegraphics[scale=1,height=12cm,width=11cm]{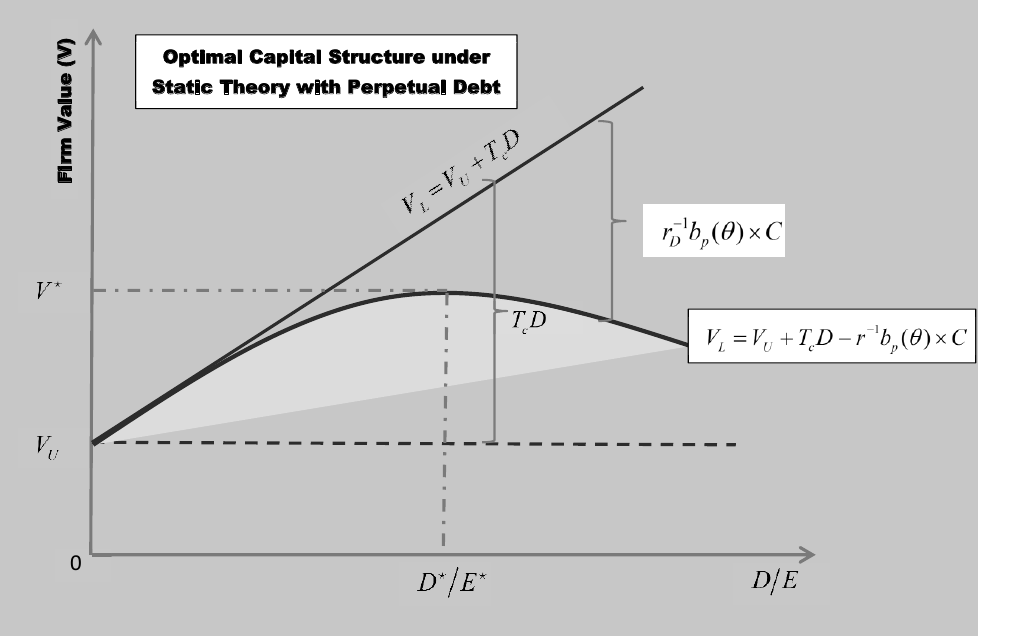}}
      \label{fig:RegulatorCareerOptionOptCapFig}
   \end{minipage}
\end{figure}
\noindent According to Kolmogorov Inequality, see, e.g., \citet[Thm.~8.3.1,~pg.~249]{AthreyaLahiri2006}, under \Cref{assum:RegSignalVolatility}, if $\varphi$ is the lower tail fraction of total tail probability, then for some threshold value $\lambda$, and assuming maximal probability thresholds, we have the nonparametric result
\begin{align}
    &b_p(\theta)=\text{Pr}\Bigl(\max_\theta\{x_1(s(\theta))\}<-\lambda\Bigr)=\varphi\frac{\text{Var}(x_1(s(\theta)))}{\lambda^2}\\
    &=\varphi\frac{\tilde{\sigma}^2_{x_1}(\theta)}{\lambda^2}>\varphi\frac{\text{Var}(x_1)}{\lambda^2}=\varphi\frac{\tilde{\sigma}^2_{x_1}}{\lambda^2}\label{eq:RegValueAtRisk}
\end{align}
where $\tilde{\sigma}^2_{x_1}(\theta)$ and $\tilde{\sigma}^2_{x_1}$ are the variances of $x_1(s(\theta))$ and $x_1$ respectively. The result in \eqref{eq:RegValueAtRisk} is important because it says that the ``large deviation probability" for firm value at risk included in $x_1(\omega_b,s(\theta))$, i.e. tail risk in bad states $\omega_b\in\Omega$, is greater than the corresponding probability if there was no regulatory signals or ``regulatory capture`` embedded in capital structure. To see this rewrite $V(s(\theta))$  in \eqref{eqgrp:FirmOptionValueWitSignal} as $V(\theta)$ so that after algebraic reduction we have
\begin{align}
   &\lim_{\sigma(\theta)\rightarrow\infty}\text{Pr}\Bigl\{\max_\theta V(\theta)<\max_\theta D(\theta)\exp\Bigl(-\lambda\sigma(\theta)\sqrt{\tau}-\Bigl(r+\tfrac{\sigma^2(\theta)}{2}\Bigr)\tau\Bigr)\Bigr\}>\varphi\frac{\tilde{\sigma}^2_{x_1}}{\lambda^2}\\
   &\Rightarrow\text{Pr}\{\max_\theta V(\theta)<0\}>\varphi\frac{1}{\lambda^2\sigma^2\tau}\text{Var}\Bigl(\ln\frac{V}{D}\Bigr),\quad |\lambda|<\left(\dfrac{\text{Var}\Bigl(\ln\frac{V}{D}\Bigr)}{\sigma^2\tau}\right)^{\frac{1}{2}}\label{eq:BankruptcyCondition}
\end{align}
where $\tilde{\sigma}^2_{x_1}=\tfrac{1}{\sigma^2\tau}\text{Var}(\ln\frac{V}{D})$. The denominator is $\sigma^2\tau$, not $D^2\tau$, because $x_1$ normalizes the log-leverage term by $\sigma\sqrt{\tau}$. The relationship in \eqref{eq:BankruptcyCondition} is equivalent to a bankruptcy condition involving leverage $\dfrac{D}{V}$. It says that if regulator signals $\theta$ induce too much volatility, i.e. uncertainty, in firm value $V(\theta)$, then the probability that the value of the firm is negative, i.e. it goes bankrupt when $V(\theta)<0$, is greater than a given threshold without regulator signals. Thus we have proven the following
\begin{prop}[Bankruptcy condition]\label{prop:RegCaptureValueAtRisk}
    Assume that regulatory capture exists so that regulators embed regulatory signals $\theta$ in firm capital structure. Then firm distance to default $x_1(\omega_b,\theta)$ in bad states of nature is greater than $x_1(\omega_b)$--the value it would be in non-regulatory capture regimes--and the probability of bankruptcy is higher than it would be in the absence of regulator signals. \hfill $\Box$
\end{prop}
\begin{rem}
   This proposition should not be interpreted to mean that in a world with no regulation, firm value-at-risk will be lower than if there is regulation. The proposition is based on an assumption of ``regulatory capture``. The absence of regulatory capture though necessary, is not a sufficient condition to support deregulation. \hfill $\Box$
\end{rem}
The leverage-vega diagram in \cref{fig:RegulatorVegaAndAvgLeverage} is a schematic implication of the model rather than a data plot. We assume that the vega for equity-linked compensation is tied to firm performance and that it is rolled over the life of the firm. The relationship between vega and firm leverage $L=\dfrac{D}{V}$ is implied by \eqref{eq:VegaBS} and \eqref{eq:VegaSG}. Thus, critical time points occur where $\vartheta(L(t|\,\sigma(\theta)))=L(t|\,\sigma(\theta))$ in vega-leverage space. Vega attains a maximum for some value of $L$. The roots of the equations in \eqref{eq:VegaBS} and \eqref{eq:VegaSG} can be obtained by superimposing a plot of $\vartheta$ and $L$ on the same schematic diagram. \cref{fig:RegulatorVegaAndAvgLeverage} provides a proof-of-concept sketch of that procedure for a root at $t^\star$.

\tab The sketch suggests that the rate of increase in firm leverage $\tfrac{dL}{dt}$, and the rate of increase of firm vega  $\tfrac{d\vartheta}{dt}$ conditioned on regulator signal dependent risk $\sigma(\theta)$, may be dispositive of bankruptcy as well. It shows that the rate of diminishing returns to vega begins \noindent just after maximum vega $\vartheta_{max}$, at $t^*$, even though the rate of firm leverage is increasing. However, beyond some critical point $(L_c,\vartheta_c)$ (a root of say \eqref{eq:VegaSG}) the firm is so highly levered that it goes bankrupt. So the period $(t^*,t_c]$ is a time when warning signals of possible bankruptcy takes place. These artifacts of vega and leverage are summarized thus:
\begin{prop}[Menu of firm response to leverage induced regulator signals]\label{prop:RegSignalMenuFirmResponse}
   Let $(L(t|\sigma(\theta))$ and  $\vartheta(L(t|\,\sigma(\theta)))$ be time dependent firm leverage and vega respectively for a given risk profile $\sigma(\theta)$ induced by regulator signal $\theta$. Critical time points are given by the solution to $\vartheta(L(t|\,\sigma(\theta)))=L(t|\,\sigma(\theta))$, and induce the following menu of firm response to leverage and regulator signal:
   \begin{enumerate}
       \item[A1] Go leverage window: $\tfrac{d\vartheta}{dt}>\tfrac{dL}{dt}$ and $L<\vartheta$\label{propitem:RegMenuGoForIt}
       \item[A2] Early warning window: $\tfrac{d\vartheta}{dt}<\tfrac{dL}{dt}$ and $L<\vartheta$\label{propitem:RegMenuWarning}
       \item[A3] Crash or bankruptcy window: $\tfrac{d\vartheta}{dt}<\tfrac{dL}{dt}$ and $L>\vartheta$\label{propitem:RegMenuCrash}
   \end{enumerate}
\end{prop}
\begin{figure}[!htb]
\centering
   \begin{minipage}[h]{0.5\linewidth}
      \captionof{figure}{Firm vega $(\vartheta)$ and leverage $L$ for given risk signal $\sigma(\theta)$ with early warning and bankruptcy windows}
      \centerline{\includegraphics[scale=.8]{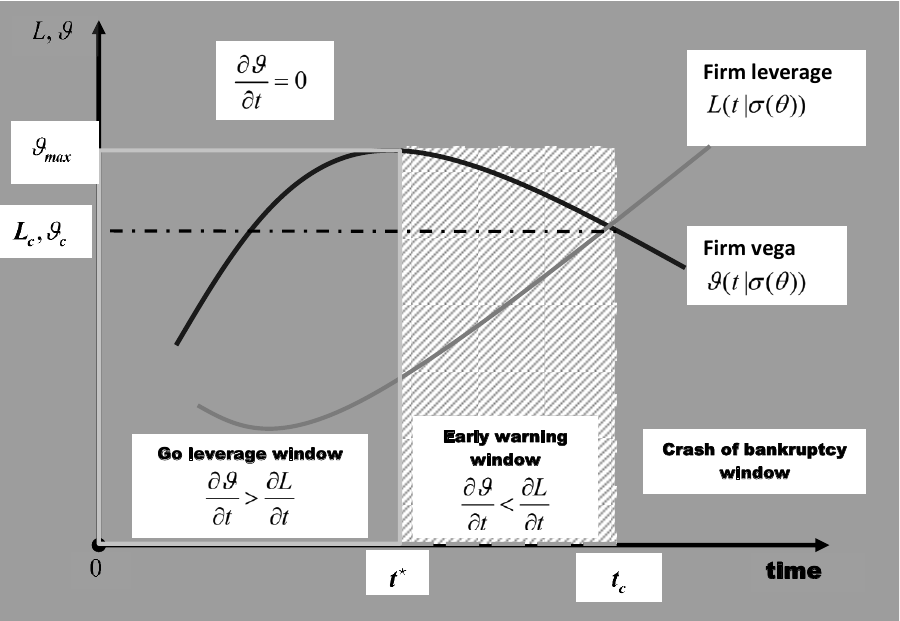}}         \label{fig:RegulatorVegaAndAvgLeverage}
   \end{minipage}
      \begin{flushleft}
         $L_c,\;\vartheta_c$ are the critical points for firm leverage and vega. Beyond that the firm's stock price crashes or it goes bankrupt. At $t^\star$ vega is at a maximum. In time period $(0,t^\star]$ leverage is encouraged. In period $(t^\star,t_c]$ leverage should proceed with caution. At time $t>t_c$ the firm is at risk for bankruptcy.
      \end{flushleft}
\end{figure}
\section{Proof-of-Concept Calibration}\label{sec:ProofOfConcept}
\tab This section provides a proof-of-concept calibration. The purpose is not to estimate the model or to fit an observed sample. The purpose is to show how the theoretical objects in \Cref{prop:RegCareerGoalFirmProfit,lem:RegulatorCareerCallOption,prop:RegCaptureValueAtRisk,prop:RegSignalMenuFirmResponse} can be evaluated in a simple numerical environment.

\tab Normalize initial firm value to $V_0=1$, set the regulator's equity-linked share to $\alpha=0.05$, and let the regulator's current compensation strike be $K=0.030$. Let a regulatory signal $\theta\geq0$ affect firm value, debt, and volatility through
\begin{align}
   V(\theta)&=V_0\{1+0.08\theta\},\\
   D(\theta)&=D_0\{1+0.15\theta\},\\
   \sigma(\theta)&=\sigma_0\{1+0.30\theta\},
\end{align}
where $D_0=0.55$ and $\sigma_0=0.20$. This calibration intentionally imposes $D_\theta/D>V_\theta/V$: debt grows at rate $0.15/(1+0.15\theta)$ while firm value grows at rate $0.08/(1+0.08\theta)$ over the grid. Hence the proof of concept is in the weak-penalty leverage channel of \Cref{assum:RegValueFunction}. The parameter $\rho$ is not separately calibrated in the option-value exercise; it is maintained as small enough that the net term $b\Pi_VV_\theta+(b\Pi_D-\rho)D_\theta$ remains positive. The regulator's career option payoff is
\begin{equation}
   C(\theta)=\bigl(\alpha V(\theta)-K\bigr)^+.
\end{equation}
The leverage ratio is $L(\theta)=D(\theta)/V(\theta)$, and the signal-adjusted Black--Scholes--Merton vega proxy is
\begin{equation}
   \vartheta(\theta)=V(\theta)\phi(x_1(\theta))\sqrt{\tau},
   \qquad
   x_1(\theta)=\frac{\log(V(\theta)/D(\theta))+(r+\sigma(\theta)^2/2)\tau}{\sigma(\theta)\sqrt{\tau}},
\end{equation}
with $r=0.03$ and $\tau=1$. These values are illustrative and are chosen only to display the comparative statics of the model.
\begin{figure}[!htbp]
\centering
\begin{subfigure}[t]{0.49\linewidth}
\centering
\includegraphics[width=\linewidth]{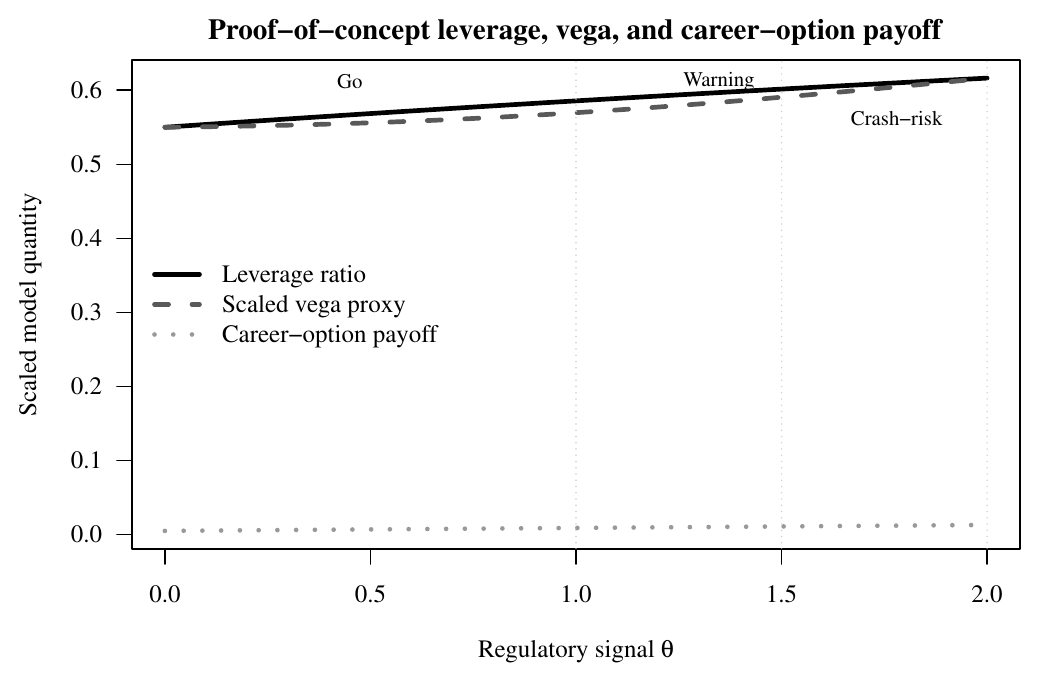}
\caption{Signal path}
\label{fig:ProofConceptLeverageVega}
\end{subfigure}\hfill
\begin{subfigure}[t]{0.49\linewidth}
\centering
\includegraphics[width=\linewidth]{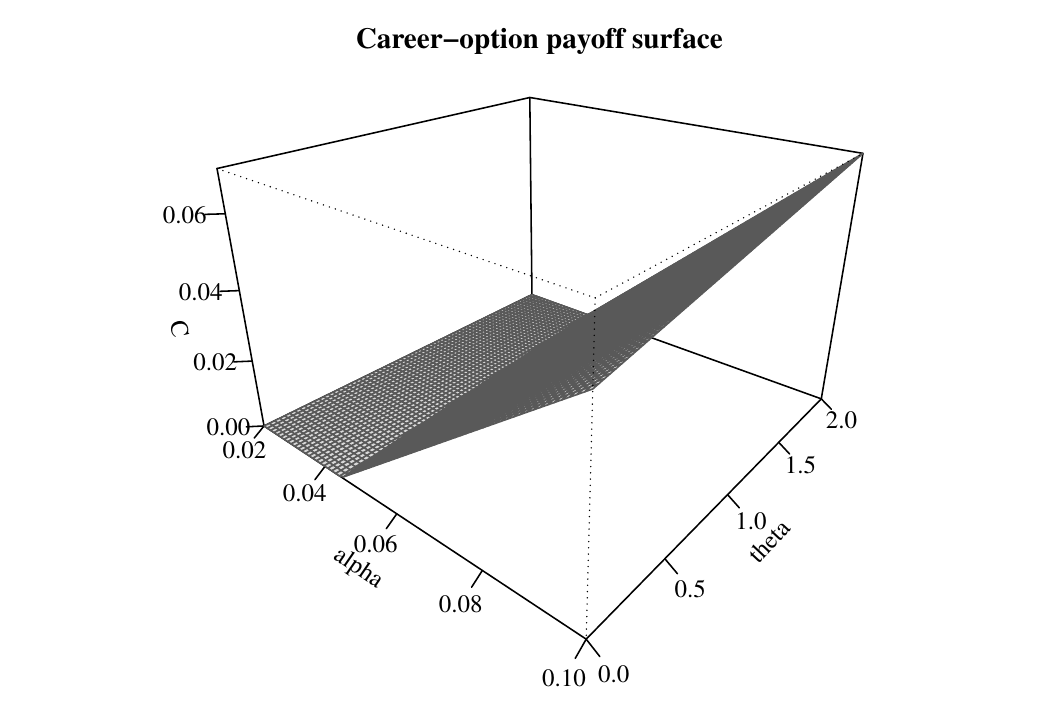}
\caption{Career-option surface}
\label{fig:ProofConceptCareerOptionSurface}
\end{subfigure}
\caption{Proof-of-concept calibration diagnostics}\label{fig:ProofConceptDiagnostics}
\begin{minipage}{0.92\linewidth}
\footnotesize
\textit{Notes:} Panel (a) plots leverage, the scaled vega proxy, and the regulator's career-option payoff along the signal path. Panel (b) plots the career-option payoff over the regulatory signal $\theta$ and the equity-linked compensation share $\alpha$. The plots are illustrative diagnostics for the model, not estimates.
\end{minipage}
\end{figure}

\tab The proof-of-concept classification follows \Cref{prop:RegSignalMenuFirmResponse}. A signal is in the go-leverage region when $d\vartheta/d\theta>dL/d\theta$ and $L<\vartheta$. It is in the early-warning region when $d\vartheta/d\theta<dL/d\theta$ and $L<\vartheta$. It is in the crash region when $d\vartheta/d\theta<dL/d\theta$ and $L>\vartheta$. In a discrete grid, the derivatives can be approximated by first differences. Thus the model can be implemented by evaluating the following table of theoretical objects.

\begin{table}[H]
\centering
\caption{Proof-of-concept leverage-vega classification}\label{tab:ProofConceptCalibration}
\begin{tabular}{lrrrrl}
\hline\hline
Signal & $V(\theta)$ & $D(\theta)$ & $L(\theta)$ & $C(\theta)$ & Region\\
\hline
$\theta=0.0$ & 1.000 & 0.550 & 0.550 & 0.020 & Baseline\\
$\theta=0.5$ & 1.040 & 0.591 & 0.568 & 0.022 & Go leverage\\
$\theta=1.0$ & 1.080 & 0.633 & 0.586 & 0.024 & Go leverage\\
$\theta=1.5$ & 1.120 & 0.674 & 0.602 & 0.026 & Early warning\\
$\theta=2.0$ & 1.160 & 0.715 & 0.616 & 0.028 & Crash-risk window\\
\hline\hline
\end{tabular}
\begin{minipage}{0.90\linewidth}
\footnotesize
\textit{Notes:} The numbers are an illustrative calibration, not estimates. The region labels summarize the qualitative inequalities in \Cref{prop:RegSignalMenuFirmResponse}. The example shows how a regulator-induced signal can raise both the value of the regulator's career option and the firm's leverage-risk exposure.
\end{minipage}
\end{table}

\tab The table illustrates the mechanism. As $\theta$ rises, firm value increases, but debt and volatility increase more aggressively. The regulator's option value rises because the payoff is convex in firm value. At the same time, leverage rises and the vega-leverage relation moves toward the warning region. The proof-of-concept shows the theoretical mechanism: signals can move a firm across leverage-vega regions, while applied work would be needed to determine whether a particular firm or industry followed that path.

\tab The preceding calibration uses a linear signal path. To show how nonlinear signals change the mechanism, consider a quadratic signal map
\begin{equation}
   s_q(\theta)=\theta+\gamma\theta^2,
   \qquad \gamma=0.40,
\end{equation}
and let firm value, debt, and volatility depend on $s_q(\theta)$ rather than directly on $\theta$:
\begin{align}
   V_q(\theta)&=V_0\{1+0.08s_q(\theta)\},\\
   D_q(\theta)&=D_0\{1+0.15s_q(\theta)\},\\
   \sigma_q(\theta)&=\sigma_0\{1+0.30s_q(\theta)\}.
\end{align}
The marginal signal effect is now state dependent, $s_q'(\theta)=1+2\gamma\theta$. Thus the option-value and leverage channels accelerate at high signal intensities. Economically, this means that a modest signal may have a small effect near $\theta=0$, while a stronger signal can move the firm rapidly into the early-warning or crash-risk region.

\tab The quadratic signal is not a Bayesian information signal. It is a reduced-form nonlinear signal-response map, similar in spirit to threshold and smooth-transition specifications used to capture nonlinear policy and state responses in applied dynamics \citep{Hansen1999Threshold,Terasvirta1994STAR,VanDijkTerasvirtaFranses2002STAR,Hamilton2016Regimes}. A Bayesian reputation-learning version, in which firms update beliefs about a regulator's type from repeated signals, is conceptually natural but would require a separate belief-state law of motion and equilibrium learning problem. We therefore leave that extension outside the present proof-of-concept calibration.

\begin{figure}[!htbp]
\centering
\includegraphics[width=0.92\linewidth]{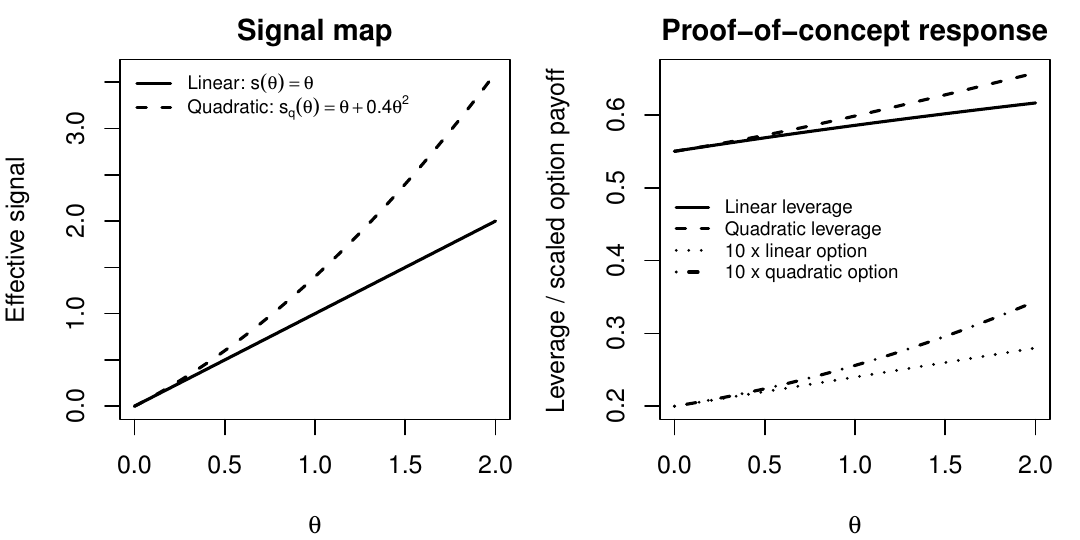}
\caption{Linear and quadratic proof-of-concept signal paths}\label{fig:ProofConceptQuadraticSignal}
\begin{minipage}{0.90\linewidth}
\footnotesize
\textit{Notes:} The left panel compares the linear signal $s(\theta)=\theta$ with the quadratic signal $s_q(\theta)=\theta+0.40\theta^2$. The right panel compares leverage and the scaled career-option payoff under the two signal maps. The quadratic signal amplifies both leverage and private option value at high $\theta$.
\end{minipage}
\end{figure}

\begin{table}[H]
\centering
\caption{Quadratic-signal proof-of-concept classification}\label{tab:ProofConceptQuadraticCalibration}
\begin{tabular}{lrrrrrl}
\hline\hline
Signal & $s_q(\theta)$ & $V_q(\theta)$ & $D_q(\theta)$ & $L_q(\theta)$ & $C_q(\theta)$ & Region\\
\hline
$\theta=0.0$ & 0.000 & 1.000 & 0.550 & 0.550 & 0.020 & Baseline\\
$\theta=0.5$ & 0.600 & 1.048 & 0.600 & 0.572 & 0.022 & Go leverage\\
$\theta=1.0$ & 1.400 & 1.112 & 0.666 & 0.598 & 0.026 & Early warning\\
$\theta=1.5$ & 2.400 & 1.192 & 0.748 & 0.628 & 0.030 & Crash-risk window\\
$\theta=2.0$ & 3.600 & 1.288 & 0.847 & 0.658 & 0.034 & Crash-risk window\\
\hline\hline
\end{tabular}
\begin{minipage}{0.90\linewidth}
\footnotesize
\textit{Notes:} The nonlinear signal uses the same primitive parameters as \cref{tab:ProofConceptCalibration}, but replaces $\theta$ by $s_q(\theta)=\theta+0.40\theta^2$. Relative to the linear calibration, the quadratic signal accelerates the movement into warning and crash-risk regions because debt, volatility, and career-option value rise more quickly at high signal intensity.
\end{minipage}
\end{table}
\subsection{Robustness and quantitative comparative statics}\label{subsec:CalibrationRobustness}
\tab The preceding tables use one normalized path. To make the proof-of-concept less dependent on a single deterministic calibration, this subsection reports five additional numerical checks. These checks are still illustrative rather than estimated. They show how the warning regions move with time-to-maturity, leverage, compensation convexity, stochastic signal realizations, and quitting horizons.

\tab First, because the vega expression in \eqref{eq:VegaSG} depends explicitly on $\tau$, the calibration is repeated for $\tau\in\{1,3,5\}$ years. Longer maturities raise the option's exposure to volatility and alter the signal intensity at which the tail-risk proxy becomes material. Second, the fixed values $\alpha=0.05$ and $D_0/V_0=0.55$ are replaced by a grid with $\alpha\in[0.01,0.20]$ and $D_0/V_0\in[0.30,0.80]$. For each grid point, the critical signal $\theta^\star$ is the first signal intensity at which the distance-to-default proxy falls below $1.5$ and the regulator's private option payoff is positive. Third, the deterministic signal is replaced by a stochastic signal draw, $\theta\sim\Gamma(2.4,0.38)$ truncated to the calibration range, with small multiplicative disturbances to firm value, debt, and volatility. This produces an empirical distribution for distance-to-default rather than a single path.

\begin{figure}[!htbp]
\centering
\begin{subfigure}[t]{0.49\linewidth}
\centering
\includegraphics[width=\linewidth]{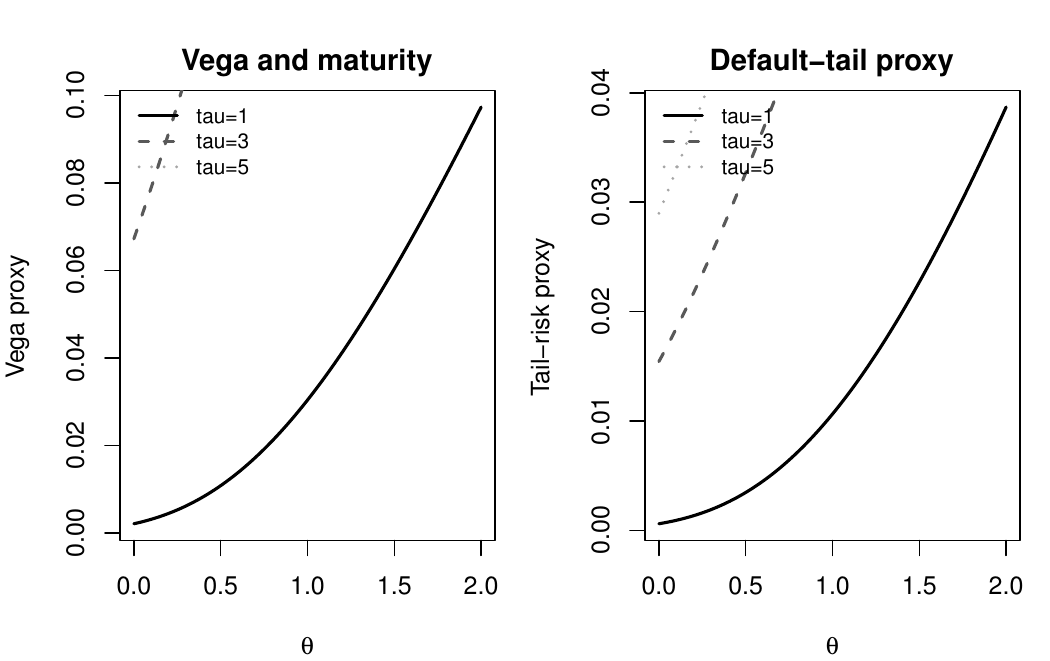}
\caption{Maturity effects}
\label{fig:ProofConceptTauDiagnostics}
\end{subfigure}\hfill
\begin{subfigure}[t]{0.49\linewidth}
\centering
\includegraphics[width=\linewidth]{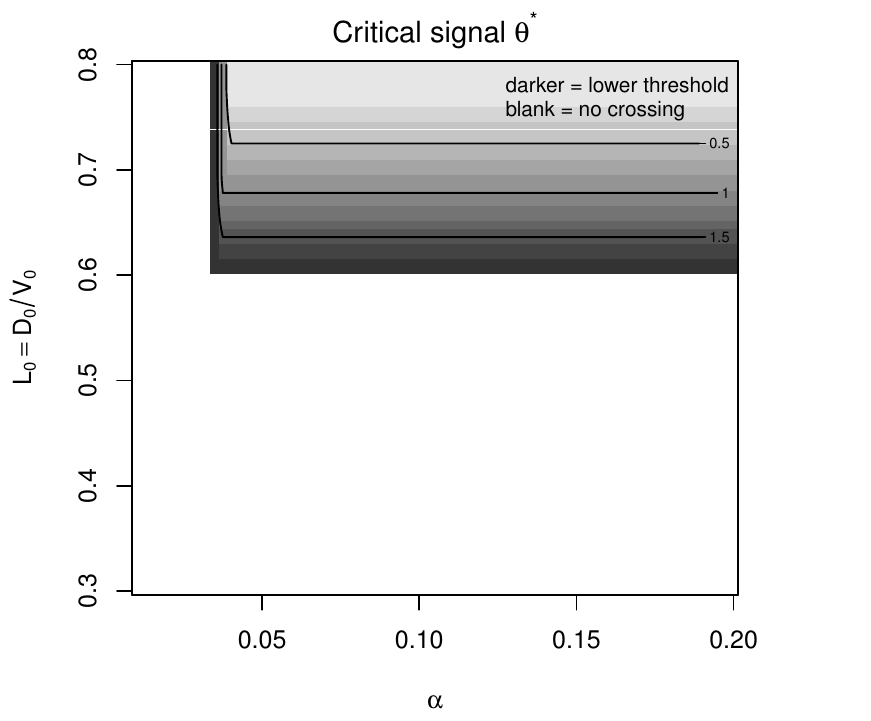}
\caption{Sensitivity grid}
\label{fig:ProofConceptSensitivityHeatmap}
\end{subfigure}
\caption{Maturity and parameter-sensitivity diagnostics}\label{fig:ProofConceptMaturitySensitivity}
\begin{minipage}{0.92\linewidth}
\footnotesize
\textit{Notes:} Panel (a) repeats the vega and tail-risk calculations for $\tau=1,3,5$. Panel (b) reports the critical signal threshold $\theta^\star$ over the equity-linked compensation share $\alpha$ and initial leverage $D_0/V_0$. Darker regions indicate lower critical thresholds. Blank regions do not cross the selected warning threshold on the grid.
\end{minipage}
\end{figure}

\tab The sensitivity grid is the most direct robustness check. It shows that the warning boundary is not a knife-edge implication of the baseline normalization. When initial leverage and volatility are closer to bank-like values, a smaller regulatory signal is sufficient to move the firm toward the warning region. Conversely, when the firm begins with low leverage or the equity-linked compensation share is small, the private option payoff and the leverage-vega warning criterion are less tightly aligned.

\tab The same strike, $K=0.030$, is used in the stopping-horizon panel. This avoids changing the payoff normalization across exercises and makes the horizon grid informative: the career option is in the money for sufficiently long horizons and sufficiently strong signals, so the plot shows how quitting horizon and signal intensity affect option value.

\begin{figure}[!htbp]
\centering
\includegraphics[width=0.88\linewidth]{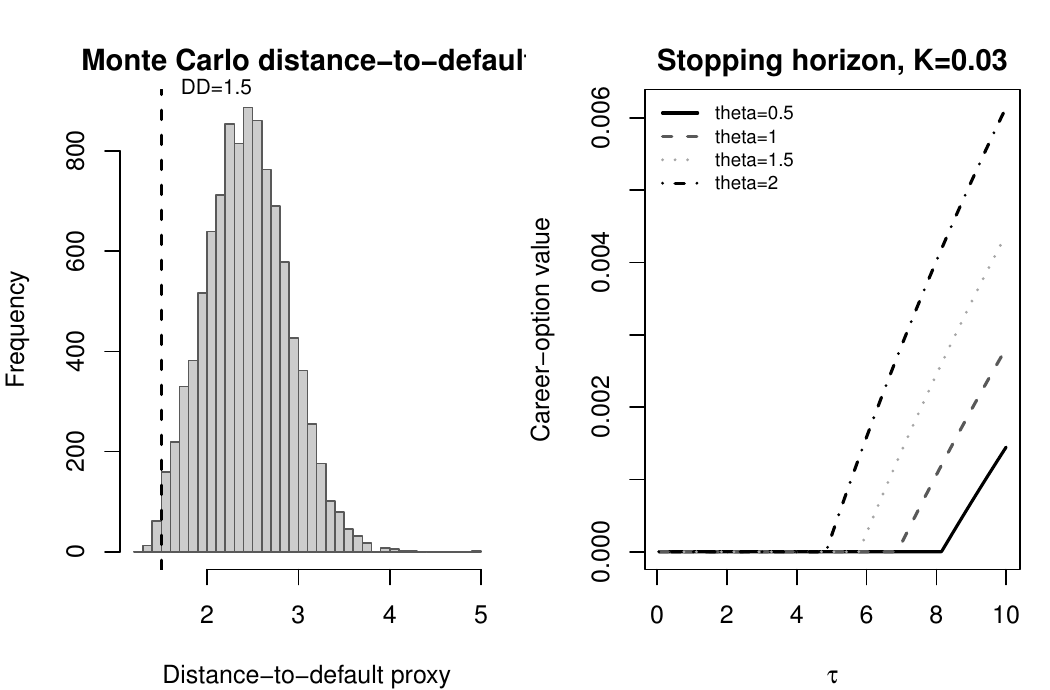}
\caption{Stochastic signal and stopping-horizon diagnostics}\label{fig:ProofConceptMonteCarloStopping}
\begin{minipage}{0.90\linewidth}
\footnotesize
\textit{Notes:} The left panel reports a Monte Carlo distribution of the distance-to-default proxy under stochastic signal realizations and noisy value, debt, and volatility paths. The dashed line marks $DD=1.5$. The right panel reports the regulator's career-option value as a function of the quitting horizon $\tau$ for several signal intensities, using the baseline strike $K=0.030$. The stopping-horizon plot is a grid diagnostic for the optimal quitting-time logic in \Cref{lem:RegulatorCareerCallOption}; it is not a full dynamic programming solution.
\end{minipage}
\end{figure}

\begin{table}[H]
\centering
\caption{Robustness calibration summary}\label{tab:ProofConceptRobustnessSummary}
\input{ProofConceptCalibrationSummary}

\begin{minipage}{0.90\linewidth}
\footnotesize
\textit{Notes:} The bank-like anchor sets $D_0/V_0=0.65$ and $\sigma_0=0.25$, values closer to large regulated financial firms than the normalized baseline. The stochastic signal column uses $10{,}000$ simulated signal draws with small noise added to value, debt, and volatility. The tail-risk proxy is $\Phi(-DD)$, where $DD$ is the distance-to-default proxy generated from $x_1$ in \eqref{eq:FirmOptionXOne}.
\end{minipage}
\end{table}

\tab The robustness exercises give three takeaways. First, the maturity parameter matters: longer option horizons change the vega profile and therefore the warning boundary. Second, industry anchoring matters: the bank-like calibration crosses the $5\%$ tail-risk threshold at a much lower signal intensity than the normalized baseline. Third, stochastic signals convert the deterministic warning-region story into a probability statement; in the reported simulation, about $0.8\%$ of paths fall below $DD=1.5$, while the upper tail of the default-risk proxy remains materially larger than the median. Thus the calibration remains a proof of concept, but it now shows how the model can be stress-tested before moving to a full empirical implementation.
\subsection{Historical regulatory-signal diagnostics}\label{subsec:HistoricalDiagnostics}
\tab Our model does not claim to forecast the exact dates of institutional failures: it is more diagnostic. It identifies settings in which the regulatory signal is absent, softened, or converted into a profit-favoring signal, while the firm moves toward high leverage, high option vega, and lower distance to default. Three recent financial episodes illustrate the three channels.

\tab Lehman Brothers is the regulatory-vacuum case. Before its 2008 failure, Lehman was not supervised by the Federal Reserve as a bank holding company. The Federal Reserve later emphasized that it was not Lehman's supervisor and that Lehman was overseen through the SEC's voluntary Consolidated Supervised Entity program rather than through Federal Reserve consolidated supervision \citep{Bernanke2010Lehman}. The SEC similarly described the CSE program as an attempt to fill a statutory gap for large investment-bank holding companies that lacked a consolidated supervisor, and concluded after 2008 that voluntary regulation did not work \citep{SEC2008CSE}. In the language of the model, the relevant signal was not an affirmative command to take risk but the absence of a binding debt and liquidity penalty. A weak or missing penalty on leverage lets the firm remain in the go-leverage region until market stress reveals that distance to default has become too small.

\tab Silicon Valley Bank is the distorted-signal case. Greg Becker, president and chief executive officer of SVB Financial Group and Silicon Valley Bank, was elected as a Class A director of the Federal Reserve Bank of San Francisco in 2019 \citep{FRBSF2019BeckerBoard}. The point is not that this board role mechanically caused the bank's failure; rather, it illustrates a conflict channel in which the regulated industry is close to the institution that supplies supervisory signals. The Federal Reserve's post-failure review found that SVB's board and management failed to manage interest-rate, liquidity, and governance risks, that supervisors did not fully appreciate the bank's vulnerabilities as it grew, and that the tailoring approach after the Economic Growth, Regulatory Relief, and Consumer Protection Act reduced standards and promoted a less assertive supervisory stance \citep{FederalReserve2023SVBReview}. The same review emphasized that SVB invested largely uninsured deposits in longer-maturity securities and became acutely exposed when interest rates rose. In the model's terms, the regulatory signal permitted a balance-sheet path in which leverage and duration exposure increased while distance to default weakened.

\tab Signature Bank is the converted-signal case. Former Representative Barney Frank, co-author of the Dodd--Frank Act, joined Signature Bank's board after leaving Congress; press coverage at the time of the 2018 regulatory rollback noted that Signature stood to benefit from the higher asset threshold for enhanced regulation and that Frank had received compensation from the bank, while Frank denied that the board role changed his policy views \citep{Stein2018FrankSignature}. The FDIC's post-failure review concluded that Signature's failure was rooted in poor management, rapid and unrestrained growth, inadequate risk management, overreliance on uninsured deposits, and vulnerability to crypto-industry deposit volatility; it also found that supervisory efforts could have been more timely, forward-looking, and forceful \citep{FDIC2023SignatureSupervision}. This is consistent with the model's converted-signal mechanism: a public-regarding regulatory stance can become a softer post-agency signal when the agent's future payoff is tied to the regulated firm's value.

\tab The cases above are illustrative; they are not proof that the model is the sole explanation for any failure. Lehman represents a missing supervisory signal, SVB a distorted and weakly enforced supervisory signal, and Signature a converted signal mediated by post-public-service incentives and light oversight. The threshold evidence in \citet{KalmenovitzVijXiao2025} gives the empirical proof for the model's key incentive channel: when the legal restriction lowers the future value of the revolving-door option, regulators behave as if that option is valuable and alter current choices to preserve it. In all three historical cases, the relevant prediction is the same: when the debt penalty is weak and the regulator's effective signal does not bind early, a firm can cross the leverage-vega warning region before the public constraint arrives.
\section{Conclusion}\label{sec:Conclusion}
\tab This paper develops a theoretical model of a regulator's career option to quit and join a regulated firm. The model links regulatory signals, firm pricing and leverage choices, and the convexity of future labor-market compensation. When the regulator's outside opportunity resembles a call option on firm value, signals that increase firm value, leverage, or volatility can increase the value of the regulator's private career option. The same signals may increase value-at-risk in bad states when they push leverage and vega into the warning or crash regions described in the model.

\tab The proof-of-concept calibration illustrates the model's comparative statics. It shows how a regulator-induced signal can raise the career-option payoff while moving the firm from a go-leverage region into an early-warning or crash-risk region. The exercise is deliberately modest: it clarifies the internal mechanism and provides a template for future quantitative implementation.

\tab The model suggests several directions for further research. One can estimate the signal process in settings where regulatory decisions and post-agency employment outcomes are observed; characterize optimal quitting times for regulators with career options; design contracts or cooling-off rules that reduce option-like incentives for regulatory capture; and test whether leverage-vega warning regions predict distress after regulatory forbearance. A Bayesian reputation-learning extension, in which firms update beliefs about regulator type from observed signals, is also natural; however, it would add a belief-state variable and updating law, so it is left outside the present mechanism. The main theoretical message is that revolving-door incentives are not merely static conflicts of interest. When future compensation is convex in regulated firm value, regulatory signals can become state-contingent financial claims embedded in the firm's capital structure.

\appendix
\renewcommand{\theHequation}{\Alph{section}.\arabic{equation}}
\addappheadtotoc
\section{Abstract legislative mechanism design for embedding regulator signals}\label{apx:AbstractMechanismDesign}

\tab The main advantage of formulating the problem in the context of an abstract legislative mechanism is it allows us to plug in select variables in the reduced form so that we can examine a taxonomy of regulator signals. In the sequel we define incentive compatibility and participation constraints for our regulator. And we characterize the revelation principle which induces a regulator [and firm] to report its type truthfully. Suffice it to say that in the sequel, the regulator designs an implicit game or indirect legislative mechanism--consistent with maximizing her welfare function--in which she frames the strategies or choices available to a firm. In order for the firm's participation constraint to be binding, the price of its product must be at least equal to its marginal cost. This is consistent with the revelation principle in \citet[pg.~913]{BaronMyerson1982}. In which case, the firm's incentive compatibility constraint rests on its ability to maximize profits in that milieu.\footnote{We thank Sam Myers Jr. for clarification on this point.} By contrast, the regulator reports her type truthfully in order to maximize her career option. The interested reader is referred to \citet{BaligaSjostrum2008} for a brief exposition or \citet{Borgers2008} for a more detailed exposition on mechanism design theory, and \citet[pg.~208,~Appendix]{DalBo2006} for details on application and implementation in the context of a revolving door model.
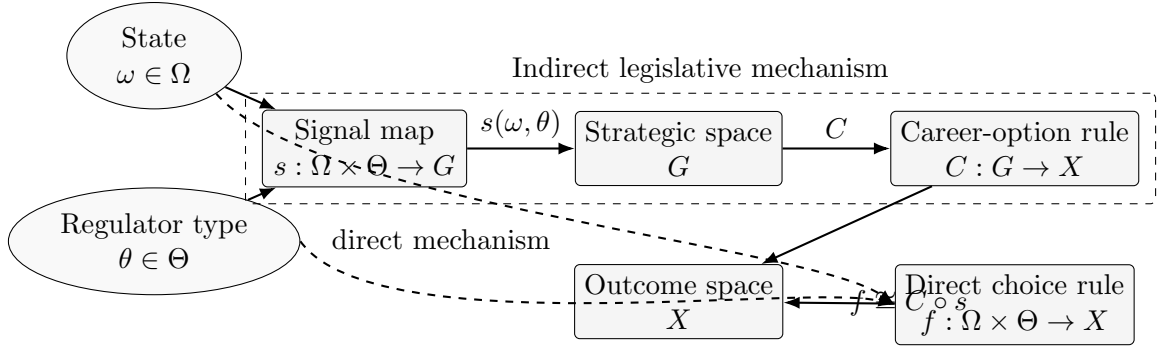
\begin{figure}[!htbp]
\centering
\resizebox{0.92\linewidth}{!}{%
\begin{tikzpicture}[
  node distance=1.05cm and 1.45cm,
  box/.style={draw=black, rounded corners=2pt, fill=gray!8, align=center, minimum width=2.75cm, minimum height=0.82cm},
  state/.style={draw=black, ellipse, fill=gray!5, align=center, minimum width=2.35cm, minimum height=0.75cm},
  arrow/.style={-{Latex[length=2.2mm]}, thick, draw=black},
  dashedarrow/.style={-{Latex[length=2.2mm]}, thick, dashed, draw=black}
]
\node[state] (omega) {State\\$\omega\in\Omega$};
\node[state, below=of omega] (theta) {Regulator type\\$\theta\in\Theta$};
\node[box, right=of $(omega)!0.5!(theta)$] (signal) {Signal map\\$s:\Omega\times\Theta\to G$};
\node[box, right=of signal] (game) {Strategic space\\$G$};
\node[box, right=of game] (career) {Career-option rule\\$C:G\to X$};
\node[box, below=1.05cm of career] (direct) {Direct choice rule\\$f:\Omega\times\Theta\to X$};
\node[box, below=1.05cm of game] (outcome) {Outcome space\\$X$};
\draw[arrow] (omega) -- (signal);
\draw[arrow] (theta) -- (signal);
\draw[arrow] (signal) -- node[above] {$s(\omega,\theta)$} (game);
\draw[arrow] (game) -- node[above] {$C$} (career);
\draw[arrow] (career) -- (outcome);
\draw[dashedarrow] (omega.south east) .. controls +(1.3,-1.5) and +(-1.1,0.8) .. node[below left] {direct mechanism} (direct.west);
\draw[dashedarrow] (theta.east) .. controls +(1.0,-1.4) and +(-1.2,0.3) .. (direct.west);
\draw[arrow] (direct) -- node[right] {$f\cong C\circ s$} (outcome);
\node[draw=black, dashed, rounded corners=2pt, fit=(signal)(game)(career), inner sep=0.22cm, label=above:{Indirect legislative mechanism}] {};
\end{tikzpicture}%
}
\caption{Regulator's abstract mechanism design}\label{fig:RegulatorCareerOptionFig01}
\begin{minipage}{0.88\linewidth}
\footnotesize
\textit{Notes:} The regulator's type and the state of nature determine a signal into the strategic space $G$. The career-option rule maps the induced strategic outcome into $X$. The dashed path is the equivalent direct choice rule, summarized by $f\cong C\circ s$.
\end{minipage}
\end{figure}\begin{defn}[Mechanism]\label{defn:Mechanism}~\
   A mechanism $\mathcal{M}=(G,C)$ defines a set of feasible strategic actions $G$ available to an agent in outcome space $X$ with outcome rule $C:G\rightarrow X$. \hfill $\Box$
\end{defn}
\tab That said, let $\Theta$ be a type space for regulators, and $\Omega$ be a sample space or states of nature, $\mathcal{F}$ be the $\sigma$-field of Borel measurable subsets of $\Omega$, and $P$ a probability measure on $\Omega$. Let $X$ be an outcome space, and $f$ be a stochastic choice function such that $f:\Omega\times\Theta\rightarrow X$ is a direct legislative mechanism. Based on her type, the regulator sends a state dependent signal $s$ into some strategic space $G$, i.e. $s:\Omega\times\Theta\rightarrow G$. For instance, the space $G$ may comprise the signal $s(\omega,\theta)$ obtained from the regulator's welfare maximization; a menu of price (p) dependent profits $\Pi(p)$ for the firm, and the strategic form $g$. So that
\begin{equation}
   G= \{s(\omega,\theta),\Pi(p),g\}
\end{equation}
However, that signal $s(\omega,\theta)$ is subject to the vagaries of the state of nature. Nonetheless, it might be favorable regulation or deregulation. On the basis of interactions in $G$, i.e. the firms response, the regulator maps a payoff function $C$ into $X$, i.e $C:G\rightarrow X$.  These mappings are depicted in \cref{fig:RegulatorCareerOptionFig01}.
\begin{defn}[(PC): Regulator's participation constraint]\label{defn:RegParticipateConstraint}~\\
   Let $K$ be the compensation package of the regulator, assumed fixed over the epoch of career mobility, and $V(\omega,\Pi(p(\theta)))$ be the value of the firm as a function of price dependent profits. The regulator will participate in exercise of a state contingent career [call] option when
   $$C(s(\omega,\theta))=(\alpha V(\omega,\Pi(p(\theta)))-K)^+$$
   for some equity stake $0<\alpha<1$ \hfill $\Box$
\end{defn}
\begin{defn}[(IC): Regulator's incentive compatibility constraint]\label{defn:RegIncentiveConstraint}~\\
   Let $\theta$ be the regulator's true type and $\hat{\theta}$ be any other reported type. Let $(\Omega,\mathcal{F},P)$ be a probability space for states of nature. Then for some event $A\in\mathcal{F}$ we have
   $$\int_AC(s(\omega,\theta))dP(\omega)\geq\int_AC(s(\omega,\hat{\theta}))dP(\omega)$$
   That is, the regulator has incentive to reveal her true type because the expected value of her career option is greater than the expected value if she does not. \hfill $\Box$
\end{defn}
\cref{fig:RegulatorCareerOptionFig01} depicts an indirect mechanism, i.e. ``topological lifting" of the stochastic choice function $f$. That is, for some $\omega\in\Omega$ we have a composite function $C\circ s$ that is functionally equivalent to $f$.
\begin{defn}[Revelation principle]\label{defn:RegRevelationPrinciple}~\\
Given any feasible career option mechanism $C:G\rightarrow X$, there exist a feasible direct legislative mechanism $f:\Omega\times\Theta\rightarrow X$ such that for $A\in\mathcal{F}$ we have
\begin{align}
    f(\omega,\theta)&\cong (C\circ s)(\omega,\theta)\label{eq:RegulateTopologyLifting}\\
    \int_Af(\omega,\theta)dP(\omega)&=\int_A(C\circ s)(\omega,\theta)dP(\omega)
\end{align}
\hfill $\Box$
\end{defn}
With full information the regulator's [options] contract can be described by the direct mechanism $f$. There is no signalling or need for the regulator to implement a set of rules for the firm. Because everyone knows that regardless of the state of nature $\Omega$, the regulator will quit and join the regulated firm for a payoff in $X$. However, that scenario may be associated with a "$P$-negligible event, i.e. small probability of occurrence for probability measure $P$ on $\Omega$, in the Borel subsets $\mathcal{F}$ of $\Omega$. However, with asymmetric information we need to design an option contract in $G$-space so that its value is $C(s(\omega,\theta)),\;\;\omega\in\Omega,\;\theta\in\Theta$. The elements of that state dependent contract must include firm value and other "observable" regulatory signaling variables. For instance, abstract signals can represent the regulator's effect on taxes, sales, reporting rules, capital treatment, or other corporate liabilities that affect reported earnings and market value.
\bibliographystyle{chicago}        
\addcontentsline{toc}{section}{References} 
\raggedright
\bibliography{RegulatorCareerOption}       
\end{document}

%% file: ProofConceptCalibrationSummary.tex
\begin{tabular}{lccc}
\hline\hline
Statistic & Baseline & Bank-like & Stochastic signal\\
\hline
Initial leverage $D_0/V_0$ & 0.55 & 0.65 & random path\\
Asset volatility $\sigma_0$ & 0.20 & 0.25 & noisy $\sigma(\theta)$\\
Critical $\theta$ for $p_d>5\%$, $\tau=1$ & -- & 0.48 & --\\
Monte Carlo $\Pr\{DD<1.5\}$ & -- & -- & 0.008\\
Monte Carlo $\Pr\{DD<1.0\}$ & -- & -- & 0.000\\
Median tail-risk proxy & -- & -- & 0.0075\\
95th percentile tail-risk proxy & -- & -- & 0.0431\\
\hline\hline
\end{tabular}